\documentclass{article}
\usepackage[utf8]{inputenc}
\usepackage{xcolor}
\usepackage[hidelinks]{hyperref}
\usepackage{amsmath,amsfonts,amsthm,amssymb,IEEEtrantools,enumerate}

\usepackage[
backend=bibtex,
firstinits,
style=alphabetic,
sorting=nyt,
date=year,
isbn=false,
maxnames=4,
minnames=4,
useprefix=true
]{biblatex}
\renewbibmacro{in:}{}
\newcommand{\ignore}[1]{}
\newcommand{\Hil}{\mathcal{H}}
\newcommand{\Kil}{\mathcal{K}}
\newcommand{\Vint}{\mathtt{V}}
\newcommand{\Wint}{\mathtt{W}}
\newcommand{\dom}{\textup{dom}}
\newcommand{\Acal}{\mathcal{A}}
\newcommand{\Bcal}{\mathcal{B}}

\newcommand{\N}{\mathbb{N}}
\newcommand{\R}{\mathbb{R}}
\newcommand{\C}{\mathbb{C}}
\newcommand{\im}{\textup{Im}}
\newcommand{\re}{\textup{Re}}
\newcommand{\Uni}{\mathcal{U}}
\newcommand{\FT}{\mathcal{F}}
\newcommand{\PV}{\mathcal{P}}

\newtheorem{theorem}{Theorem}[section]
\newtheorem{lemma}[theorem]{Lemma}
\newtheorem{proposition}[theorem]{Proposition}
\newtheorem{corollary}[theorem]{Corollary}
\theoremstyle{definition}
\newtheorem{definition}[theorem]{Definition}
\newtheorem{example}[theorem]{Example}
\theoremstyle{remark}
\newtheorem{remark}[theorem]{Remark}
\newtheorem*{remark*}{Remark}

\title{Relative Positions of Half-sided Modular Inclusions}
\author{Ian Koot \\ Friedrich-Alexander-Universität Erlangen-Nürnberg, Germany \\ Email: \href{mailto:ian.koot@fau.de}{ian.koot@fau.de}}
\date{March 23, 2025}

\begin{document}

\maketitle

\begin{abstract}
Let $K_1 \subset H$ and $K_2 \subset H$ be half-sided modular inclusions in a common standard subspace $H$. We prove that the inclusion $K_1 \subset K_2$ holds if and only if we have an inclusion of spectral subspaces of the generators of the positive one-parameter groups associated to the half-sided modular inclusions $K_1 \subset H$ and $K_2 \subset H$. From this we give a characterization of this situation in terms of (operator valued) symmetric inner functions. We illustrate these characterizations with some examples of non-trivial phenomena occurring in this setting.
\end{abstract}

\section{Introduction}
Tomita-Takesaki modular theory has become an important tool in the analysis of quantum field theory: it has found application through its close relation to KMS states \cite[Ch. 13]{Takesaki1970}, to separation properties of field theories \cite{Buchholz1990a}, in the context of the Bisognano-Wichmann theorem \cite{BRUNETTI2002} and to black hole physics \cite{Witten2022}, among others (see for example \cite{Borchers2000} for a review of some applications). In the context of modular theory and quantum field theory, a situation of prominent interest is the half-sided modular inclusion. An inclusion of von Neumann algebras $\Acal \subset \Bcal$ is half-sided modular if they have a common standard vector $\Omega$ and satisfy 
	\[ \Delta_{\Acal,\Omega}^{-it} \Bcal \Delta_{\Acal,\Omega}^{it} \subset \Bcal \quad \text{for } t \geq 0. \]
(here $t \mapsto \Delta_{\Acal,\Omega}^{it}$ is the modular group of $\Acal$ with respect to $\Omega$, see Example \ref{ex:VNA}). Physically, such inclusions appear for example in the context of light cones in thermal states \cite{Borchers1999}, light cones in vacuum states for massless theories \cite{Buchholz1978}, subregions of Rindler wedges \cite{Borchers1996, Morinelli_2022}, or spacetimes with a horizon \cite{Jefferson2019, Leutheusser2023}. In addition, half-sided modular inclusions have turned out to be instrumental in the construction and analysis of chiral conformal field theories \cite{Longo2010,Wiesbrock1993a, Guido1998} and establishment of energy inequalities like the QNEC \cite{Ceyhan2020, Morinelli_2022, Hollands2025}.

In the above-mentioned examples, modular theory is applied to a von Neumann algebra $\Acal$ and standard vector $\Omega$. The most general setting in which modular theory can be applied, however, is the setting of standard subspaces: closed, real subspaces $H$ inside a complex Hilbert space $\Hil$ such that $H \cap iH = \{0\}$ and $\overline{H  + iH} = \Hil$. The standard subspace $H_{\Acal, \Omega} := \overline{\Acal_{sa}\Omega}$ associated to $\Acal$ and $\Omega$ is then just one particular example. By focussing on those properties essential to modular theory, results in the theory often become clearer when they are formulated in terms of standard subspaces. On top of this, standard subspaces that do not necessarily come from a von Neumann algebra have been applied to the construction of QFT models \cite{Lechner2014, CorreadaSilva2023} and in the investigation into entropy of coherent excitations \cite{Ciolli2022,Longo2024}.

As inclusions of operator algebras play an important role in (algebraic) quantum field theory, there is also a strong interest in inclusions of standard subspaces and their relation to modular theory and representation theory (e.g. \cite{BRUNETTI2002, Neeb2021, Ciolli2022, Adamo2024}). In many cases where half-sided modular inclusions appear, the structure of inclusions takes the form of many smaller algebras (or standard subspaces) that lie half-sidedly in some fixed larger algebra corresponding to a (local) environment region; examples of this include light cones in the quasi-local algebra \cite{Borchers1999} or null cuts inside a Rindler wedge \cite{Morinelli_2022}. Although the representation theory of a single half-sided modular inclusion is relatively simple (due to the Stone-von Neumann uniqueness theorem, see Theorem \ref{thm:SvN}), how different half-sided modular inclusions relate to each other is not so obvious. In particular, the possible relative positions of standard subspaces $K_1$ and $K_2$ that form a half-sided modular inclusion within a common standard subspace $H$ are currently unknown.

In this article we characterize when two standard subspaces $K_1$ and $K_2$ that lie as half-sided modular inclusions in a common environment standard subspace $H$ form an inclusion. We do this by switching perspective to standard pairs, which are pairs $(H,U)$ of a positively generated one-parameter group $U$ and a standard subspace $H$ such that $U(s)H \subset H$ for $s \geq 0$. This structure is known to be equivalent to a half-sided modular inclusion, and two half-sided modular inclusions $K_1 \subset H$ and $K_2 \subset H$ satisfy $K_1 \subset K_2$ if and only if their associated standard pairs $(H,U_1)$ and $(H,U_2)$ satisfy $U_1(1)H \subset U_2(1)H$. We show that $K_1 \subset K_2$ is equivalent to an inclusion of spectral subspaces of the generators of $U_1$ and $U_2$ (Theorem \ref{thm:main}), translating the problem from an inclusion of standard subspaces (which can be quite subtle, see for example \cite{FiglioliniGuido} and \cite{CorreadaSilva2023}) to an inclusion of complex subspaces, which is much simpler. An earlier investigation in the context of von Neumann algebras into the relation between inclusions $U_1(1)H \subset U_2(1)H$ and the unitary one-parameter groups $U_1$ and $U_2$ is contained in \cite{BORCHERS1997}. There, the inclusion is characterized by an inequality of the semigroups $s \mapsto U_j(is)$ for $s \geq 0$, $j = 1,2$. Our methods result in a simpler proof of this characterization, but also expand the results significantly; a comparison between our results and those in \cite{BORCHERS1997} can be found in Remark \ref{rem:Borchers}.

We then apply the representation theory of standard pairs to this result, and thereby show that $K_1 \subset K_2$ if and only if there is a representation where $K_1$ and $K_2$ differ by an operator-valued inner function (Proposition \ref{prop:SymInnChar}). In the case where $K_1 \subset H$ and $K_2 \subset H$ induce irreducible standard pairs, this reduces to $K_1 = \varphi(\ln \Delta_H)K_2$ for a symmetric inner function $\varphi: \C_- \rightarrow \C$ (Corollary \ref{cor:HSMIIrredVersion}). We also use these characterizations to provide examples that illustrate some nontrivial phenomena that can occur (Example \ref{ex:nontriv} and \ref{ex:GenIneqCounter}). The characterizations in terms of analytic functions are somewhat reminiscent of the characterizations of endomorphisms of standard pairs given in \cite{Longo2010}. In Remark \ref{rem:LW} we explain the differences and similarities.

This paper is organized as follows. In Section \ref{sec:prereq} we introduce the topics required for the results in this paper, and adapt some standard results to our context. In Section \ref{sec:MainResult} we prove the equivalence of the inclusion $U_1(1)H \subset U_2(1)H$ of the standard pairs $(H,U_1)$ and $(H,U_2)$ to the inclusion of spectral subspaces of the generators of $U_1$ and $U_2$, which is the main result of this paper (Theorem \ref{thm:main}). In Section \ref{sec:applications} we use the representation theory of the canonical commutation relations to prove the characterizations in terms of (operator valued) symmetric inner functions. 

We only consider the mathematical structure of half-sided modular inclusions, and give no direct physical applications of the structural results in this paper, leaving them to be discussed elsewhere. In general we expect there to be many applications, given the generality of the results and the many closely related physical applications as explained above. 

\subsubsection*{Acknowledgements}
I would like to thank G. Lechner for his many helpful comments and suggestions during the writing of this article.

\section{Modular theory, standard subspaces and their inclusions}
\label{sec:prereq}
We first introduce some of the concepts that play a role in the main results for the convenience of the reader and to fix some notation. A more comprehensive exposition on these topics can be found in \cite[Ch. 3]{LongoSibiu}, \cite{Neeb2017}, \cite[Ch. 2]{BratteliRobinson1}, \cite[Ch. 5.2]{Bratteli1987} and \cite{Derezinski2006}. We advise the more experienced reader to continue to the main results starting in Section \ref{sec:MainResult}.

 We introduce some general notation:
 	\begin{IEEEeqnarray*}{rClrCl}
 	\mathbb{C}_+ & := & \{ z \in \C \mid 0 < \im(z) \} \quad & (M[f]\psi)(\theta) & :=& f(\theta)\psi(\theta) \\
 	\mathbb{C}_- & := & \{z \in \C \mid \im(z) < 0\} \quad & (C[\phi]\psi)(\theta) & := & (\phi * \psi)(\theta) \\
 	\mathbb{S}_{\alpha} & := & \{ z \in \C \mid 0 < \im(z) < \alpha \} \quad & (\FT \psi)(\lambda) & := & \frac{1}{\sqrt{2\pi}}\int_\R \psi(\theta)e^{-i\lambda\theta} \, d\theta
 	\end{IEEEeqnarray*}
In addition, throughout this article $P_j$ (respectively $P$) will be the self-adjoint generator of the strongly continuous one-parameter group $U_j$ for $j = 1,2$ (respectively $U$).

\subsection{Standard subspaces}

\begin{definition}
Let $\Hil$ be a (complex) Hilbert space. A real subspace $H \subset \Hil$ is called \textbf{cyclic} if $H + iH$ is dense in $\Hil$ and \textbf{separating} if $H \cap iH = \{0\}$. It is called \textbf{standard} if it is cyclic and separating. If $H \subset \Hil$ is standard, the operator
	\[ S_H: H + iH \rightarrow H + iH, \quad h_1 + ih_2 \mapsto h_1 - ih_2  \]
is called the \textbf{Tomita operator}.
\end{definition}
For any set $A \subset \Hil$ we write $A' = \{ \psi \in \Hil \mid \forall a \in A : \im\langle \psi, a \rangle = 0 \}$ for the symplectic complement. It is an easy exercise to show that $H$ is cyclic (resp. separating) if and only if $H'$ is separating (resp. cyclic). In particular, if $H$ is standard, then $H'$ is also standard.
\begin{theorem}[Tomita-Takesaki theorem, \cite{Takesaki1970} {\cite[Ch. 3]{LongoSibiu}}]
\label{thm:Tomita}
Let $H \subset \Hil$ be a standard subspace.
\begin{enumerate}[a)]
\item The map $S_H$ is a closed anti-linear operator, with $S_H^* = S_{H'}$.
\item The map $S_H$ has polar decomposition $S_H = J_H \Delta_H^{\frac{1}{2}}$ where $J_H$ is an anti-unitary involution and $\Delta_H$ is a positive operator such that $\ker \Delta_H = \{0\}$, and we have $J_{H'} = J_H$ and $\Delta_{H'} = \Delta_{H}^{-1}$.
\item We have
	\[ J_H H = H', \quad \Delta_H^{it}H = H \text{ for all } t \in \R \]
and $J_H \Delta_H J_H = \Delta_H^{-1}$.
\end{enumerate}
\end{theorem}
\begin{definition}
The operators $J_H$ and $\Delta_H$ are called the \textbf{modular conjugation} and the \textbf{modular operator} of $H$, respectively. The unitary one-parameter group $t \mapsto \Delta_H^{it}$ is called the \textbf{modular group} of $H$.
\end{definition}
Clearly, we can recover the standard subspace $H$ from $S_H$ by considering $\{ \psi \in \dom\,S_H \mid S_H\psi = \psi\}$. In this way one can show that there is actually a one-to-one correspondence between closed anti-linear involutions and standard subspaces \cite[Prop. 3.2]{LongoSibiu}. On the other hand, given an anti-unitary involution $J$ and a positive operator $\Delta$ such that $J\Delta J = \Delta^{-1}$, one can define a closed anti-linear involution by $S = J\Delta^\frac{1}{2}$, and this also turns out to be a one-to-one correspondence \cite[Cor. 3.5]{LongoSibiu}. We therefore see that a standard subspace is completely determined by its modular objects. In Examples \ref{ex:ex1} and \ref{ex:ex2} we make use of this fact to define some standard subspaces.
\begin{example}
\label{ex:ex1}
Let $\Hil_0 = L^2(\R, d\theta)$ be the $L^2$ space with respect to the Lebesgue measure. We construct a standard subspace $H_0 \subset \Hil_0$ by specifying its modular objects, namely
	\[ (\Delta_{H_0}^{it}\psi)(\theta) := \psi(\theta - 2 \pi t), \quad (J_{H_0}\psi)(\theta) := \overline{\psi(\theta)} \]
so that
	\[ H_0 := \{ \psi \in \dom\, \Delta_{H_0}^{\frac{1}{2}} \mid \Delta_{H_0}^{\frac{1}{2}}\psi = J_{H_0}\psi \}. \]
More concretely, in \cite[Lemma 4.1]{Lechner2014} it is shown that $\psi \in \dom\, \Delta_{H_0}^{\frac{1}{2}}$ if and only if it is the boundary value (in terms of the $L^2$-norm) of an analytic function on $\mathbb{S}_{\pi}$, such that $\psi(\theta + iy) \in L^2(\R,d\theta)$ for all $0 < y < \pi$ (that is, $\psi$ is included in the Hardy space $H^2(\mathbb{S}_{\pi})$). If $\psi \in H_0$, then in addition to the above, one has $\psi(\theta + \pi i) = \overline{\psi(\theta)}$ for almost all $\theta \in \R$. Note that the von Neumann algebra generated by the $\Delta_{H_0}^{it}$ is maximally abelian, as can clearly be seen in Example \ref{ex:ex2}.
\end{example}
\begin{example}
\label{ex:ex2}
By Fourier transforming all the objects in  \ref{ex:ex1}, we can construct a standard subspace $\widetilde{H}_0 := \FT H_0 \subset \Hil_0$ with modular data
	\[ (\Delta_{\widetilde{H}_0}^{it}\psi)(\lambda) = e^{-2\pi t \lambda i} \psi(\lambda), \quad (J_{\widetilde{H}_0}\psi)(\lambda) = \overline{\psi(-\lambda)}\]
so that $\psi \in \widetilde{H}_0$ if and only if
	\[ \int_{-\infty}^\infty e^{-2 \pi \lambda}|\psi(\lambda)|^2 \, d\lambda < \infty \quad \text{and} \quad e^{-\pi \lambda}\psi(\lambda) = \overline{\psi(-\lambda)}. \]
We note that in this realisation it is clear that the von Neumann algebra $\{\Delta_{\widetilde{H}_0}^{it} \mid t \in \R\}'' = L^\infty(\R)$ is maximally abelian (i.e. is equal to its own commutant). This means in particular that any operator that commutes with all $\Delta_{\widetilde{H}_0}^{it}$ is of the form $f(\ln \Delta_{\widetilde{H}_0})$ for some measurable function $f: \R \rightarrow \R$.
\end{example}

We now make a brief remark on general inclusions of standard subspaces $K \subset H$. In principle, this is a question of extensions of closed operators $S_K \subset S_H$. However, as is typical for modular theory, we can translate the problem to a problem in complex analysis, which we do in Proposition \ref{prop:InclusionChar}. For this, we need a notion of analyticity for functions that take values in $B(\Hil)$. There are \emph{a priori} three obvious ways to interpret the derivative of such a function (in the norm topology, strong operator topology and weak operator topology) but we note that by \cite[Sec. 3.9]{hille1948functional} they are all equivalent. In this paper we will therefore not make a distinction and simply write `analytic'. Whenever we write `so-continuous' we mean continuous is the strong operator topology.

\begin{proposition}[{\cite[Thm. 2.5]{Borchers1999a}}, {\cite[Thm 2.12]{ARAKI2005}}]
\label{prop:InclusionChar}
Let $K, H \subset \Hil$ be standard subspaces. Then $K \subset H$ if and only if the map
	\[ F: \R \rightarrow B(\Hil), \quad t \mapsto \Delta_H^{-it}\Delta_K^{it} \]
extends to a bounded so-continuous function on $\overline{\mathbb{S}_{\frac{1}{2}}}$, analytic in $\mathbb{S}_{\frac{1}{2}}$ such that
	\[ F(t + \tfrac{i}{2}) = \Delta_H^{-it}J_HJ_K \Delta_K^{it}, \quad t \in \R. \]
\end{proposition}

\begin{example}
\label{ex:VNA}
A very important example of standard subspaces are those constructed from a von Neumann algebra $\Acal \subset B(\Hil)$ and a standard (i.e. cyclic and separating) vector $\Omega \in \Hil$ for $\Acal$. One defines 
	\[ H_{\Acal,\Omega} := \overline{\{ A\Omega \mid A \in \Acal, A^* = A \}} \]
We will denote the modular objects of $H_{\Acal, \Omega}$ as $\Delta_{\Acal,\Omega}^{it}$ and $J_{\Acal,\Omega}$.

It should be noted that this assignment `forgets' some information about the von Neumann algebra, as one can construct different von Neumann algebras on the same Hilbert space with the same standard vector that have the same modular objects (see e.g. \cite[Thm 3.5]{Buchholz2010}). However, given some von Neumann algebra $\Acal \subset B(\Hil)$ (which one can consider as some kind of environment) and standard vector $\Omega \in \Hil$, it is the case that any sub-von Neumann algebras $\Bcal_1,\Bcal_2 \subset \Acal$ so that $\Omega$ is also standard for $\Bcal_1$ and $\Bcal_2$ satisfy $\Bcal_1 \subset \Bcal_2$ if and only if $H_{\Bcal_1,\Omega} \subset H_{\Bcal_2,\Omega}$ \cite[Prop. 3.24 (a)]{LongoSibiu}. In other words, for any given von Neumann algebra $\Acal \subset B(\Hil)$ and standard vector $\Omega \in \Hil$ for $\Acal$, the map
	\[ \{ \Bcal \subset B(\Hil) \text{ VNA}\mid \Bcal \subset \Acal, \, \Omega \text{  std. for } \Bcal \} \rightarrow \{ H \subset H_{\Acal,\Omega} \mid H \text{ std. subsp.} \} \]
given by $\Bcal \mapsto H_{\Bcal,\Omega}$ is injective.
\end{example}

\subsection{Half-sided modular inclusions and standard pairs}
Next we introduce the connection between half-sided modular inclusions and standard pairs, starting with the latter.
\begin{definition}
A \textbf{standard pair} is a pair $(H, U)$ of a standard subspace $H \subset \Hil$ and a positively generated strongly continuous one-parameter group $U: \R \rightarrow \mathcal{U}(\Hil)$ such that
	\[ U(s)H \subset H \quad \text{ for } s \geq 0. \]
A standard pair is called \textbf{non-degenerate} if $U$ has no invariant vectors other than 0. 
\end{definition}
In particular, a standard pair induces a family of inclusions of standard subspaces, namely $U(s)H \subset H$ for all $s \geq 0$. The above definition is also sometimes referred to as a \emph{positive} standard pair, to differentiate it from the same definition but for $s \leq 0$, which is then called a \emph{negative} standard pair. In this article we will only be dealing with positive standard pairs.
\begin{theorem}[Borchers' theorem, {\cite[Thm. II.9]{Borchers1992}} {\cite[Thm. 1]{Florig1998}}]
\label{thm:Borchers}
Let $(H,U)$ be a standard pair. Then
	\[ \Delta_H^{it}U(s)\Delta_H^{-it} = U(e^{-2\pi t}s) \quad \text{and} \quad J_H U(s) J_H = U(-s)\]
\end{theorem}
This result means that the inclusions $U(s)H \subset H$ for $s \geq 0$ have the following interesting property: for $t \geq 0$ one sees that
	\[ \Delta_H^{-it}U(s)H = U(e^{2\pi t}s)H = U(s)U(( e^{2\pi t} - 1)s)H \subset U(s)H  \]
since $(e^{2\pi t} - 1)s \geq 0$ for $t\geq 0$ and $s \geq 0$. We can formulate this property using only modular theory:
\begin{definition}
Let $K \subset \Hil$ and $H \subset \Hil$ be standard subspaces, and $K \subset H$. This inclusion is called a \textbf{half-sided modular inclusion} if
	\[ \Delta_H^{-it} K \subset K \quad \text{ for all } t \geq 0. \]
We call a half-sided modular inclusion $K \subset H$ \textbf{non-degenerate} if
	\[ \bigcap_{t \geq 0} \Delta_H^{-it}K = \{0\} \]
\end{definition}
As with the definition of a standard pair, one sometimes wants to differentiate positive and negative half-sided modular inclusions; for this article, we will only use the definition given above.

As remarked above, from Borchers' theorem one can easily see that standard pairs define half-sided modular inclusions. What is surprising, is that this process can be reversed:
\begin{theorem}[{\cite[Thm. 3]{Wiesbrock1993}} {\cite[Thm. 2.1]{ARAKI2005}}]
\label{thm:Wiesbrock}
Let $K \subset H \subset \Hil$ be a half-sided modular inclusion. Then the identity
	\[ U(1 - e^{-2\pi t}) := \Delta_K^{it}\Delta_H^{-it} \]
can be extended to a positively generated one-parameter group $U: \R \rightarrow \mathcal{U}(\Hil)$ such that $(H,U)$ is a standard pair and $K = U(1)H$.
\end{theorem}
So we see that there is a one-to-one correspondence between standard pairs and half-sided modular inclusions. However, it is not always straightforward to concretely calculate one structure when given the other; in particular one can in principle check whether an inclusions of standard subspaces $K \subset H$ is a half-sided modular inclusion by only having concrete information about $\Delta_H^{it}$ (as is for example the case when considering a standard vector associated to a KMS-state on a von Neumann algebra) and not actually calculating $\Delta_K^{it}$, which is in general hard. For this reason, it is good to know that non-degeneracy of a half-sided modular inclusion (which can be checked knowing only the interaction between $K$ and $\Delta_H^{it}$) and non-degeneracy of standard pairs (which can be checked knowing only the one-parameter group $U$) are actually equivalent, as the name suggests:
\begin{proposition}
\label{prop:nondeg}
A half-sided modular inclusion $K \subset H$ is non-degenerate if and only if its associated standard pair (as in Theorem \ref{thm:Wiesbrock}) is non-degenerate.
\end{proposition}
\begin{proof}
Let $(H,U)$ be the standard pair associated to $K \subset H$, and let $\pi$ be the projection onto $\ker P$, where $P$ is the generator of $U$. We first prove that $H \cap \ker P = \pi H$.

We note that by Theorem \ref{thm:Borchers} for $\psi \in \ker P$ we have 
\[ U(s)\Delta_H^{it}\psi = \Delta_H^{it}U(e^{2\pi t}s)\psi = \Delta_H^{it}\psi \]
and $U(s)J_H\psi = J_HU(-s)\psi = J_H\psi$. So $\Delta_H^{it}$ and $J_H$ commute with $\pi$. For $\varphi, \psi \in \dom \, \Delta_H^{\frac{1}{2}}$ we note that the analytic functions $z \mapsto \langle \Delta_H^{i\overline{z}}\psi, \pi \varphi \rangle$ and $z \mapsto \langle \pi \psi, \Delta_H^{-iz}\varphi \rangle$ on $\mathbb{S}_{\frac{1}{2}}$ have the same boundary values on $\R$, so they must be equal. In particular we get $\langle \Delta_H^{\frac{1}{2}} \psi, \pi \varphi \rangle = \langle \pi \psi, \Delta_H^{\frac{1}{2}} \varphi \rangle$, so because $\Delta_H^{\frac{1}{2}}$ is self-adjoint, we have
\[\pi \, \dom \, \Delta_H^{\frac{1}{2}} = \dom\, \Delta_H^{\frac{1}{2}} \cap \ker P.\]
Because $\pi$ commutes with $J$, we then have $H \cap \ker P = \pi H$.

We also note that because $U(s)H$ is a decreasing family of standard subspaces and $U(1)H = K$, we have that
	\[ \bigcap_{s \in \R} U(s)H = \bigcap_{s \geq 0} U(s)H = \bigcap_{t \geq 0} \Delta_H^{-it}K. \]

Now first suppose $\ker P \neq \{0\}$. Because $H + iH$ is dense in $\Hil$, we have that $\pi H \neq \{0\}$. This means that
	\[ \{0 \} \neq \pi H = H \cap \ker P \subset \bigcap_{s \geq 0} U(s)H = \bigcap_{t \geq 0} \Delta_H^{-it}K   \]
so $K \subset H$ is degenerate.

Next, we suppose that we have nonzero $h \in \bigcap_{t \geq 0} \Delta_H^{-it}K$. Then $U(s)h \in H$ for all $s \in \R$, so for all $h' \in H'$ we have that $\langle h', U(s)h \rangle = \langle U(s)h,h'\rangle$ for all $s \in \R$. So because $U$ has a positive generator, the function
	\[ s \mapsto \langle h', U(s)h\rangle = \langle h, U(-s) h'\rangle \]
can be extended to an entire bounded function, so it is constant. This means that for all $s \in \R$ and all $\psi \in H' + i H'$, we have $\langle \psi, U(s) h\rangle = \langle \psi,h\rangle$. Since $H' + iH'$ is dense, this means that $h \in \ker P$.
\end{proof}
\begin{example}
\label{ex:ex1ctd}
We can enhance the standard subspace defined in Example \ref{ex:ex1} by the unitary one-parameter group defined by
	\[ (U_0(s)\psi)(\theta) = e^{ie^\theta s}\psi(\theta).  \]
Clearly, $U_0$ is positively generated. Since the function
	\[ \R \ni \theta \mapsto (U_0(s)\psi)(\theta + iy) = e^{i\cos(y)e^{\theta}s}e^{-\sin(y)e^{\theta}s}\psi(\theta) \]
lies in $L^2(\R)$, we have that $U_0(s)H_0 \subset H_0$ for all $s \geq 0$.
\end{example}
\begin{example}
\label{ex:ex2ctd}
For $s \in \R\setminus \{0\}$, we see that
	\[ T_s := \frac{1}{\sqrt{2\pi}}\FT[\theta \mapsto e^{ise^{\theta}}] = \frac{1}{2}\delta + \frac{1}{2\pi} \PV\left(e^{i\lambda \ln(-is)}\Gamma(-i\lambda)\right) \]
where $\PV$ denotes the principal value at $\lambda = 0$, and the Fourier transform is taken in a distributional sense; a proof of this fact can be found in Appendix \ref{app:distroFourier}. Because $M[e^{ise^\theta}]$ is unitary for all $s \in \R$, this means that $C[T_s]: C_c^\infty(\R) \rightarrow C_c^\infty(\R)$ can be $L^2$-continuously extended to a unitary map $\widetilde{U}_0(s): L^2(\R) \rightarrow L^2(\R)$, which satisfies $\widetilde{U}_0(s) = \FT U_0(s) \FT^*$. From this it is clear that $(\widetilde{H}_0, \widetilde{U}_0)$ is a standard pair (here $\widetilde{H}_0$ is defined in Example \ref{ex:ex2}).
\end{example}
\begin{example}
\label{ex:VNActd}
Given a von Neumann algebra $\Acal \subset B(\Hil)$ with standard vector $\Omega$, and a von Neumann algebra $\Bcal \subset \Acal$, we call the inclusion $\Bcal \subset \Acal$ a half-sided modular inclusion if $\Omega$ is standard for $\Bcal$ and
	\[ \Delta_{\Acal,\Omega}^{-it}\Bcal \Delta_{\Acal,\Omega}^{it} \subset \Bcal,\quad t \geq 0. \]
As we noted in Example \ref{ex:VNA}, since both $\Bcal$ and $\Delta_{\Acal,\Omega}^{-it}\Bcal \Delta_{\Acal,\Omega}^{it}$ lie inside $\Acal$, they form an inclusion if and only if their standard subspaces form an inclusion; so given an inclusion $\Bcal \subset \Acal$, it is a half-sided modular inclusion if and only if $H_{\Bcal, \Omega} \subset H_{\Acal, \Omega}$ is a half-sided modular inclusion.  Given a von Neumann algebra $\Acal$ and standard vector $\Omega$, we therefore have an injective map
	\[ \{ \Bcal \mid \Bcal \subset \Acal \text{ HSMI} \} \rightarrow \{ H \mid H \subset H_{\Acal,\Omega} \text{ HSMI} \}. \]
Also in this restricted setting, this map is not necessarily surjective. If $H \subset H_{\Acal,\Omega}$ is a half-sided modular inclusion, then the associated standard pair $(H_{\Acal,\Omega},U)$ gives a unitary $U(1)$ such that $H = U(1)H_{\Acal,\Omega}$; this clearly is the standard subspace associated to the von Neumann algebra $U(1)\Acal U(-1)$. What is \emph{a priori} not clear, however, is whether $U(1)\Acal U(-1) \subset \Acal$. In \cite{Davidson1996} a sufficient condition is given for this to be the case.
\end{example}

\subsection{Representation theory of standard pairs}
Standard pairs are closely related to representations of the canonical commutation relations. To see how we can retrieve such a representation from standard pairs, let $(H,U)$ be a non-degenerate standard pair and $P$ the generator of $U$. By Theorem \ref{thm:Borchers}, we know that
	\[ \exp(i s \Delta_H^{it}P\Delta_H^{-it}) = \Delta_H^{it}U(s)\Delta_H^{-it} = U(e^{-2\pi t}s) = \exp(ie^{-2\pi t}s P).\]
Because the generator of a unitary one-parameter group is unique, we have that $ \Delta_H^{it}P\Delta_H^{-it} = e^{-2\pi t}P$. Since $U$ is non-degenerate, $P$ does not have $0$ as an eigenvalue, and therefore by the functional calculus we can take the logarithm on both sides. This gives $\Delta_H^{it}\ln(P)\Delta_H^{-it} = \ln(P) - 2 \pi t$, or looking at the associated one-parameter group,
	\[ \Delta_H^{it}P^{is}\Delta_H^{-it} = e^{-2\pi ts i}P^{is}. \]
We can recognise this as the Weyl form of the canonical commutation relations (taking derivatives gives, at least formally, the form $[\ln \Delta_H, \ln P] = 2\pi i$). By the Stone-von Neumann uniqueness theorem, up to unitary equivalence there is only one irreducible representation of the Weyl relations, namely the one associated to the standard pair defined in Example \ref{ex:ex1} and \ref{ex:ex1ctd} (or the unitarily equivalent \ref{ex:ex2} and \ref{ex:ex2ctd}):
\begin{example}
\label{ex:ex1ctd2}
The generator for $U_0$ defined in Example \ref{ex:ex1ctd} is the multiplication operator $P_0 = M[e^{\theta}]$, so that $\ln P_0 = M[\theta]$. On the other hand, taking the derivative of the action $(\Delta_{H_0}^{it}\psi)(\theta) = \psi(\theta - 2\pi t)$ in $t$ gives $\ln \Delta_{H_0} = 2\pi i \partial_{\theta}$. Any operator commuting with both taking derivatives as well as multiplication operators, must be constant, so we see that this representation is irreducible. This representation of the canonical commutation relations is called the Schrödinger representation.
\end{example}
\begin{example}
\label{ex:ex2ctd2}
The generator for $\widetilde{U}_0$ defined in Example \ref{ex:ex2ctd} can best be calculated as the Fourier transform of $P_0$ in Example \ref{ex:ex1ctd}. Since $\FT M[\theta]  \FT^* = i\partial_\lambda $, so $\dom \,\widetilde{P}_0 = H^2(\mathbb{S}_1)$ and $(\widetilde{P}_0\psi)(\lambda) = \psi(\lambda + i)$. That this is a positive operator can also explicitly be seen from the contour shift
	\begin{align*}
	\langle \psi, \widetilde{P}_0 \psi \rangle & = \int_\R \overline{\psi(\lambda)} \psi(\lambda + i) \, d\lambda = \int_\R \overline{\psi(\lambda + i/2)} \psi(\lambda + i/2) \, d\lambda \\
	& = \|\psi( \cdot + i/2) \|_{2}^2
	\end{align*}
\end{example}
Any non-degenerate standard pair is therefore up to unitary equivalence a multiple of the Schrödinger representation. What remains is to fit the modular conjugation into this picture. In Theorem \ref{thm:SvN} we realise the representation as a tensor product between the irreducible representation and a multiplicity space; we then show that the modular conjugation splits as $J_{H_0}$ on the irreducible factor and an anti-linear involution $J_\Kil$ on the multiplicity space. We do this by guessing $J_\Kil$, and then showing that we are only a unitary away from the actual anti-linear involution. 
\begin{theorem}[Stone-von Neumann uniqueness theorem, \cite{Neumann1931} {\cite[Cor. 5.2.15]{Bratteli1987}} {\cite[Thm. 4]{Derezinski2006}}]
\label{thm:SvN}
Let $(H,U)$ be a non-degenerate standard pair inside the Hilbert space $\Hil$. Then there exists a Hilbert space $\Kil$, an antiunitary involution $J_{\Kil}$ on $\Kil$ and a unitary map $\Vint: \Hil \rightarrow \Hil_0 \otimes \Kil$ such that
	\[ \Vint \Delta_H^{it} \Vint^* = \Delta_{H_0}^{it} \otimes 1, \quad \Vint J_H \Vint^* = J_{H_0}\otimes J_{\Kil}, \quad \Vint U(s) \Vint^* = U_0(s) \otimes 1 \]
where $\Hil_0, \Delta_{H_0}^{it}, J_{H_0}$ and $U_0(s)$ are defined as in Examples \ref{ex:ex1} and \ref{ex:ex1ctd}.
\end{theorem}
\begin{proof}
Let $P$ be the (positive) generator of $U$, and $P_0$ of $U_0$. Then $\Delta_H^{it}$ and $P^{is}$ satisfy the Weyl relations, so by the Stone-von Neumann theorem {\cite[Thm. 4]{Derezinski2006}} there exists a Hilbert space $\Kil$ and a unitary map $\Vint: \Hil \rightarrow \Hil_0 \otimes \Kil$ such that 
	\[ \Vint \Delta_H^{it} \Vint^* = \Delta_{H_0}^{it} \otimes 1 \quad \text{and} \quad \Vint P^{is} \Vint^* = P_0^{is} \otimes 1.  \]
By the functional calculus, the latter implies $\Vint U(s) \Vint^* = U_0(s) \otimes 1$ for all $s \in \R$. We note that the commutant of the representation generated by $\Delta_{H_0}^{it}$ and $P_0^{is}$ is $\C \cdot 1_{\Hil_0}$ because of irreducibility, so if $A \in B(\Hil)$ commutes with all $\Delta_H^{it}$ and $P^{is}$, then there must be a $A_{\Kil} \in B(\Kil)$ such that $\Vint A \Vint^* = 1 \otimes A_{\Kil}$. By Theorem \ref{thm:Borchers} we know that $J_HU(s)J_H = U(-s)$ and by Theorem \ref{thm:Tomita} we have that $J_H\Delta_H^{it}J_H = \Delta_H^{it}$. Similarly, $J_{H_0}U_0(s)J_{H_0} = U_0(-s)$ and $J_{H_0}\Delta_{H_0}^{it}J_{H_0} = \Delta_{H_0}^{it}$, so for any anti-unitary involution $J_\Kil$ the map $(J_{H_0} \otimes J_\Kil)\Vint J_H \Vint^*$ is a unitary operator that commutes with both $\Delta_{H_0}^{it} \otimes 1$ and $U_0(s) \otimes 1$ for $t,s \in \R$, meaning that there is a unitary operator $U_\Kil \in \Uni(\Kil)$ such that
	\[ (J_{H_0} \otimes J_{\Kil})\Vint J_H \Vint^* = 1 \otimes U_\Kil  \quad \text{ so } \quad \Vint J_H \Vint^* = J_{H_0} \otimes J_{\Kil}U_{\Kil}. \]
Because $J_H$ is an involution, we must have that $J_\Kil U_\Kil$ is also an involution.
\end{proof}
\begin{remark*}
Note that any two anti-unitary involutions are unitarily equivalent: if $J_1$ and $J_2$ are anti-unitary involutions, then $U := J_1J_2$ is a unitary, which satisfies $J_1 U = J_2 = J_2^* = U^*J_1$, so $J_1 U J_1 = U^*$. Let $f: \{ z\in \C \mid |z| = 1\} \rightarrow [-\pi, \pi)$ be the argument function, so that $U = e^{if(U)}$. We also have $J_1f(U)J_1 = f(U^*) = -f(U)$. So define $V = e^{if(U)/2}$; then $J_1VJ_1 = V^*$ and $V^2 = U$, meaning that $V^*J_1V = J_1V^2 = J_2$. This means that we can modify the $\Vint$ we get in Theorem \ref{thm:SvN} so that we can choose $J_\Kil$ to be any anti-unitary involution on $\Kil$.
\end{remark*}
\begin{remark*}
From the characterization in Theorem \ref{thm:SvN} we can read off some general properties of standard pairs. For example, for any non-degenerate standard pair $(H,U)$ the generator $P$ of $U$ has no eigenvalues, because the generator of $U_0$ has none. Also, we note that the representation is irreducible if and only if the von Neumann algebra generated by $\Delta_H^{it}$ is maximally abelian. More specifically, we have the following:
\end{remark*}
\begin{corollary}
\label{cor:multiplicity}
If $(H,U_1)$ and $(H,U_2)$ are non-degenerate standard pairs, than their induced representations (as in Theorem \ref{thm:SvN}) have the same multiplicity.
\end{corollary}
\begin{proof}
We identify $L^2(\R) \otimes \Hil \cong L^2(\R,\Hil)$ for all Hilbert spaces $\Hil$. If $\Kil_1$ is the multiplicity space of $(H,U_1)$ and $\Kil_2$ for $(H,U_2)$, we have by Theorem \ref{thm:SvN} a unitary $\Vint: L^2(\R, \Kil_1) \rightarrow L^2(\R, \Kil_2)$ such that $\Vint \Delta_{\widetilde{H}_0}^{it} \otimes 1_{\Kil_1} \Vint^* = \Delta_{\widetilde{H}_0}^{it} \otimes 1_{\Kil_2}$. Therefore 
	\[ \Vint \left(\chi_{[a,b]} \otimes 1_{\Kil_1} \right) \Vint^* = \chi_{[a,b]} \otimes 1_{\Kil_2} \]
so $\Vint$ is given by multiplication with a function $\R \rightarrow B(\Kil_1,\Kil_2)$. However, because $\Vint$ is unitary, it must be unitary in almost each fibre, meaning that $\Kil_1 \cong \Kil_2$.
\end{proof}

\section{Spectral characterization of inclusions of standard pairs}
\label{sec:MainResult}
The main result of this paper is the following characterization of inclusions of standard pairs by the inclusion of corresponding spectral subspaces:
\begin{theorem}
\label{thm:main}
Let $H \subset \Hil$ be a standard subspace, and $(H,U_1)$ and $(H,U_2)$ be non-degenerate standard pairs. Let also $E_j$ be the spectral measure associated to $P_j$ for $j=1,2$. Then
	\[ E_1[(0,1)] \leq E_2[(0,1)] \quad \Leftrightarrow \quad U_1(1)H \subset U_2(1)H \]
\end{theorem}
\begin{remark}
\label{rem:Borchers}
Before proving this result, we compare it to the characterization of inclusions of half-sided modular inclusions in \cite{BORCHERS1997}. There it is proven, that for standard pairs of von Neumann algebras $(\Acal, U_1)$ and $(\Acal, U_2)$ (i.e. $U_j$ is positively generated and $U_j(s)\Acal U_j(-s) \subset \Acal$ for $s \geq 0$) that $U_1(1)\Acal U_1(-1) \subset U_2(1) \Acal U_2(-1)$ if and only if $e^{-P_1 t} \leq e^{-P_2 t}$ for all $t \geq 0$ (\cite[Thm. 3.5]{BORCHERS1997}). The corresponding result for standard subspaces follows from Lemma \ref{lem:PLUniform}, together with the argument in Theorem \ref{thm:CtoStd} that uniform boundedness in the upper half plane implies analyticity and so-continuity. Compared to \cite{BORCHERS1997}, our methods have the advantage of relying only on 1-dimensional complex analysis, instead of 2-dimensional complex analysis (similar to the way \cite[Thm. 1]{Florig1998} simplifies the proof from \cite[Thm. II.9]{Borchers1992}). It should also be noted that the standard subspace version of Borchers' result might be strictly stronger, because as remarked in Example \ref{ex:VNActd}, given a von Neumann algebra $\Acal$ with standard vector $\Omega$ there might be half-sided modular inclusions in $H_{\Acal,\Omega}$ that do not come from half-sided modular inclusions in $\Acal$. 

Moreover, the characterization in terms of spectral projections in Theorem \ref{thm:main} allows one to more concretely work with these types of inclusions of standard subspaces, as an inclusion of complex subspaces is in general more tractable than an inclusion of standard subspaces. For example, because of Theorem \ref{thm:main} we can show that, although an inequality of the generators is implied by an inclusion of standard pairs (Corollary \ref{cor:GenIneq}), it is not sufficient (see Example \ref{ex:GenIneqCounter}). Also, through the characterization of inclusions of half-sided modular inclusions by symmetric inner functions (Corollary \ref{cor:HSMIIrredVersion}) we can construct a concrete example of  $K_1 \subset H$ and $K_2 \subset H$ half-sided modular inclusions, where $K_1 \subset K_2$ is \emph{not} a half-sided modular inclusion (Example \ref{ex:nontriv}). Such an example is especially interesting in the context of the question raised in \cite{BORCHERS1997} whether it is inherent to quantum field theory that the standard pairs that correspond to localization regions commute. In any case, we conclude that Theorem \ref{thm:main} significantly improves upon the results in \cite{BORCHERS1997}.
\end{remark}

Let us now turn to the proof of Theorem \ref{thm:main}. We will treat the `$\Leftarrow$' direction (Theorem \ref{thm:StdtoC}) and the `$\Rightarrow$' direction (Theorem \ref{thm:CtoStd}) separately, as the proofs are quite different. On a technical level, however, both directions are based on a version of Proposition \ref{prop:InclusionChar} specifically adapted to standard pairs, namely Lemma \ref{lem:StdPairExtension}. For this we need a slight generalization of well-known result in complex analysis:
\begin{lemma}
\label{lem:PLUniform}
Let $A: \overline{\mathbb{S}_{\frac{1}{2}}} \rightarrow B(\Hil)$ be a bounded so-continuous map that is analytic on $\mathbb{S}_\frac{1}{2}$ and $A_0 \in B(\Hil)$ such that
	\[ \lim_{t \rightarrow -\infty} A(t) = \lim_{t \rightarrow -\infty}  A(t+\tfrac{i}{2}) = A_0 \]
in the strong operator topology. Then $\lim_{t \rightarrow -\infty} A(t + iy) = A_0$ in the strong operator topology, uniformly in each seminorm for all $0 \leq y \leq \frac{1}{2}$.
\end{lemma}
The proof is a straightforward application of the maximum principle, which also holds for vector-valued analytic functions; it can be found in Appendix \ref{app:PLUniform}.
\begin{lemma}
\label{lem:StdPairExtension}
Let $H \subset \Hil$ be a standard subspace, and $(H,U_1)$ and $(H,U_2)$ be standard pairs. Then $U_1(1)H \subset U_2(1)H$ if and only if the map
	\[ F: \R \rightarrow B(\Hil), \quad s \mapsto U_2(-s)U_1(s)  \]
extends to a bounded so-continuous function on $\overline{\C_+}$ that is analytic on $\C_+$.
\end{lemma}
\begin{proof}
Let $\Delta_j^{it}$ be the modular group associated to the standard subspace $U_j(1)H$ for $j = 1,2$. We note that by Theorem \ref{thm:Borchers}
	\begin{equation}
	\label{eq:StdPairExtension1}
	\Delta_2^{-it}\Delta_1^{it} = \Delta_2^{-it}\Delta_H^{it}\Delta_H^{-it}\Delta_1^{it} = U_2(1 - e^{2\pi t})U_1(e^{2\pi t} - 1)
	\end{equation}
and because $J_jJ_H = U_j(2)$ we have
	\begin{align}
	\Delta_2^{-it}J_2J_1\Delta_1^{it} & = \Delta_2^{-it}J_2J_HJ_HJ_1\Delta_1^{it} = \Delta_2^{-it}U_2(2)U_1(-2)\Delta_1^{it} \nonumber\\
	& = U_2(2e^{2\pi t})\Delta_2^{-it}\Delta_1^{it}U_1(-2e^{2\pi t}) \label{eq:StdPairExtension2} \\
	& = U_2(e^{2\pi t} + 1)U_1(-e^{2 \pi t} - 1). \nonumber
	\end{align}
In addition, we note that
	\[ g:  \mathbb{S}_{\frac{1}{2}} \rightarrow \C_+, \quad z \mapsto e^{2\pi z} - 1\]
is a biholomorphic map, with inverse $g^{-1}(z) = \frac{1}{2\pi} \ln(z + 1)$ (here we use the complex logarithm with its branch cut on the negative part of the real axis). This map extends to a homeomorphism $g: \overline{\mathbb{S}_{\frac{1}{2}}} \rightarrow \overline{\C_+}\setminus \{-1\}$ through
	\[ g(t) = e^{2\pi t} - 1, \quad g(t + \tfrac{i}{2}) = - e^{2 \pi t} - 1  \text{ for }t \in \R\]
and
	\[ g^{-1}(t) = \begin{cases}
	\frac{1}{2\pi} \ln(t + 1) \quad &\text{ for } t > -1 \\
	\frac{1}{2\pi} \ln(-(t+1)) + \frac{i}{2} \quad & \text{ for } t < -1
	\end{cases} \]
If we call $\widetilde{F}(t) := \Delta_2^{-it}\Delta_1^{it}$, then we see that $\widetilde{F}(t) = F(g(t))$ for $t \in \R$. 

First suppose that $F$ extends to a bounded so-continuous function on $\overline{\C_+}$ that is analytic in $\C_+$. Then $\widetilde{F} = F \circ g$ is a bounded so-continuous map on $\overline{\mathbb{S}_\frac{1}{2}}$ analytic on $\mathbb{S}_\frac{1}{2}$ (these properties are preserved by composition), so by Proposition \ref{prop:InclusionChar} we have $U_1(1)H \subset U_2(1)H$.

Next, suppose $U_1(1)H \subset U_2(1)H$. Then by Proposition \ref{prop:InclusionChar} $\widetilde{F}$ has a so-continuous and bounded extension to $\overline{\mathbb{S}_{\frac{1}{2}}}$ which is analytic in the interior and satisfies $\widetilde{F}(t + \frac{i}{2}) = \Delta_2^{-it}J_2J_1\Delta_1^{it}$. Since $F(s) = \widetilde{F}(g^{-1}(s))$ for all $s \in \R \setminus \{-1\}$, the function $F|_{\R \setminus \{-1\}}$ has a bounded so-continuous extension to $\overline{\C_+}\setminus \{-1\}$ that is analytic on $\C_+$ (namely $\widetilde{F} \circ g^{-1}$). In addition, because $U_1$ and $U_2$ are strongly continuous and bounded, we have that $F$ is continuous on $\R$. What remains to be shown, is that this means that $\widetilde{F} \circ g^{-1}$ can be so-continuously extended to $-1$.

By \eqref{eq:StdPairExtension1} and \eqref{eq:StdPairExtension2}, together with so-continuity and boundedness of $U_1$ and $U_2$, we know that
	\[ \lim_{t \rightarrow -\infty} \widetilde{F}(t) = \lim_{t \rightarrow -\infty} \widetilde{F}(t + \tfrac{i}{2}) = U_2(1)U_1(-1) \]
in the strong operator topology. Lemma \ref{lem:PLUniform} then guarantees for all $\psi \in \Hil$ and $\varepsilon > 0$ that there is a $t_0 \in \R$ such that for all $t < t_0$ and $0 \leq y \leq \frac{1}{2}$ we have
	\[ \| \widetilde{F}(t + iy)\psi - U_2(1)U_1(-1)\psi \| < \varepsilon\]
So for all $z \in \overline{C_+}$ such that $|-1 - z| < e^{t_0}$ we have that $\|F(z)\psi - U_2(1)U_1(-1)\psi\| < \varepsilon$. So $F$ is also so-continuous in $-1$.
\end{proof}

We now come to the first direction of the proof of Theorem \ref{thm:main}, namely that the inclusion of standard subspaces $U_1(1)H \subset U_2(1)H$ implies the inclusion of complex subspaces $E_1[(0,1)]\Hil \subset E_2[(0,1)]\Hil$ (or, equivalently, the inequality of projections $E_1[(0,1)] \leq E_2[(0,1)]$). The idea of the proof is to use the result in Lemma \ref{lem:StdPairExtension} to bound a function in the upper half plane, which we already know in general to be bounded in the lower half plane, meaning that it must be constant.
\begin{theorem}
\label{thm:StdtoC}
Let $H \subset \Hil$ be a standard subspace, and $(H,U_1)$ and $(H,U_2)$ be non-degenerate standard pairs. If $U_1(1)H \subset U_2(1)H$, then
	\[0 < a_1 \leq b_1 \leq a_2 \leq b_2 \Rightarrow E_1[(a_1,b_1)] \perp E_2[(a_2,b_2)]\]
where $E_j$ is the spectral measure associated to $P_j$ for $j = 1,2$. In particular, we have
	\[ E_1[(0,1)] \leq E_2[(0,1)]. \]
\end{theorem}
\begin{proof}
Suppose $\psi_i \in E_i[(a_i,b_i)]\Hil$ for $i = 1,2$, and let $P_i$ be the (positive) generator of $U_i$. By Lemma \ref{lem:StdPairExtension} we know that
	\[ f(s) = \langle \psi_2, U_2(-s)U_1(s) \psi_1 \rangle \]
has a bounded analytic extension to $\C_+$. Because $\psi_i \in \dom \,e^{yP_i}$ for all $y \in \R$, the function $z \mapsto \langle e^{i\overline{z}P_2}\psi_2, e^{izP_1}\psi_1 \rangle$ is entire analytic, and it coincides with $f$ on the real line, so by the Schwarz reflection principle, they agree on $\C_+$. In addition, one has for $t \geq 0$ that
	\[ \|e^{tP_1}\psi_1\| \leq e^{tb_1}\|\psi_q\|, \quad \|e^{-tP_2}\psi_2\| \leq e^{-ta_2}\|\psi_2\| \]
so for $z = x + iy \in \C_-$ we have
	\[ |\langle e^{i\overline{z}P_2}\psi_2, e^{izP_1}\psi_1 \rangle| \leq \|e^{i\overline{z}P_2}\psi_2\| \|e^{izP_1}\psi_1\| \leq e^{y(a_2 - b_1)}\|\psi_1\|\|\psi_2\|.\]
So if we examine the map $z \mapsto \langle e^{i\overline{z}P_1}\psi_1, e^{izP_2}\psi_2 \rangle$ we see that it is entire, bounded on $\C_+$, and if $a_2 \geq b_1$, also bounded on $\C_-$: therefore, it is constant. In fact, for $a_2 > b_1$, we see that as $y \rightarrow -\infty$, we must have $|\langle e^{i\overline{z}P_2}\psi_2, e^{izP_1}\psi_1 \rangle| \rightarrow 0$; since the latter was constant, it must be zero. In particular, if we take $z = 0$, we see that $\langle \psi_2, \psi_1 \rangle = 0$. Seeing as $\psi_1$ and $\psi_2$ were arbitrary in $E_1[(a_1,b_1)] \Hil$ and $E_2[(a_2,b_2)] \Hil$, respectively, we get the desired result.

In particular, $E_1[(0,x)]$ is orthogonal to $E_2[(x,\infty)]$ (by taking the so-limit as $a_1$ goes to 0, $b_1$ and $a_2$ go to $x$ and $b_2$ goes to $\infty$). We note that by Theorem \ref{thm:SvN} $P_1$ and $P_2$ have no eigenvalues, so $E_j[\{x\}] = 0$ for all $x \in \R$. This all means that
	\[ E_1[(0,x)]\Hil \subset \left( E_2[(x,\infty)] \Hil\right)^\perp = E_2[(0,x)]\Hil \]
for all $x > 0$. So, taking $x = 1$, we get the desired inequality.
\end{proof}

We now turn to the reverse statement. This boils down to showing that $e^{yP_2}e^{-yP_1}$ is uniformly bounded in $y$ for $y \geq 0$. Let us briefly look at the idea behind the proof of this fact. When we take a vector $\psi$ in $E_1[(\lambda - \delta, \lambda + \delta)]\Hil$, it is close to being an eigenvector for $\lambda$. Applying $e^{-yP_1}$ to it will roughly scale the vector $\psi$ by $e^{-y\lambda}$. If $E_1[(0,\lambda + \delta)] \leq E_2[(0,\lambda + \delta)]$, then in the worst case scenario, by applying $e^{yP_2}$ to $\psi$ the norm grows with a factor of $e^{y(\lambda + \delta)}$. So if $\delta$ is small enough, the norm doesn't grow too much. We formalize this as follows:
\begin{proposition}
\label{prop:ProjImpliesUnifBdd}
Let $H \subset \Hil$ be a standard subspace, $(H,U_1)$ and $(H,U_2)$ be non-degenerate standard pairs and $E_j$ the spectral measure of the generator of $U_j$. If $E_1[(0,1)] \leq E_2[(0,1)]$, then $U_2(-z)U_1(z)$ is a well-defined bounded operator for all $z \in \overline{\C_+}$, and
	\[ \|U_2(-z)U_1(z)\| \leq 1 \quad \text{ for } z \in \overline{\C_+}\]
\end{proposition}
\begin{proof}
We write $P_j$ for the generator of $U_j$. We note first that because $\Delta_H^{it}P_j \Delta_H^{-it} = e^{-2\pi t} P_j $ we have
	\[ \Delta_H^{it}E_j[(0,1)] \Delta_H^{-it} = E_j[(0,e^{2\pi t})] \]
for $j = 1,2$. In particular, if $E_1[(0,1)] \leq E_2[(0,1)]$ then $E_1[(0,a)] \leq E_2[(0,a)]$ for all $a > 0$. This means that $E_1[(0,a)]\Hil \subset \dom(e^{yP_2})$ for all $a, y > 0$, and therefore, the operator $e^{yP_2}e^{-yP_1}$ is densely defined for all $y \geq 0$. For $z = x + iy$ with $y > 0$ we have
	\[ \| U_2(-z)U_1(z)\| = \|e^{(-ix + y)P_2}e^{(ix - y)P_1}\| =  \|e^{yP_2}e^{-yP_1}\| \]
and therefore it suffices to uniformly bound $e^{yP_2}e^{-yP_1}$.

So let now $y \geq 0$, and $\psi \in \dom(e^{y P_2}e^{-yP_1})$. We will bound $\|e^{y P_2}e^{-yP_1}\psi\|^2$ by splitting the spectrum of $P_1$ in intervals of width $\delta > 0$ and decomposing $\psi$ with respect to this splitting. So let $\delta > 0$; for readability we define $E_1^{(j)} := E_1[(j\delta, (j+1)\delta)]$, and we calculate
	\begin{IEEEeqnarray*} {Rcl}
	\IEEEeqnarraymulticol{3}{l}{
	\|e^{yP_2}e^{-yP_1}\psi\|^2 =  \sum_{j,k = 0}^\infty \langle e^{yP_2}e^{-yP_1} E_1^{(j)}\psi, e^{yP_2}e^{-yP_1} E_1^{(k)} \psi \rangle }\\
	= &   \sum_{j = 0}^\infty \sum_{n = 1}^\infty & \big( \langle e^{yP_2}e^{-yP_1} E_1^{(j)}\psi, e^{yP_2}e^{-yP_1} E_1^{(j+n)} \psi \rangle \\
	& & \quad + \langle e^{yP_2}e^{-yP_1} E_1^{(j+n)}\psi, e^{yP_2}e^{-yP_1} E_1^{(j)} \psi \rangle \big) \\
	& \IEEEeqnarraymulticol{2}{l}{ + \sum_{j = 0}^\infty \|e^{yP_2}e^{-yP_1}E_1^{(j)}\psi\|^2 } .
	\end{IEEEeqnarray*}
Because $\re\,{z} \leq |z|$ and the Cauchy-Schwarz inequality we have
	\begin{align*}
	\|e^{yP_2}e^{-yP_1}\psi\|^2 \leq & \, 2\sum_{j=0}^\infty \sum_{n=1}^\infty \|e^{2yP_2}e^{-yP_1}E_1^{(j)}\psi\|\|e^{-yP_1}E_1^{(j+n)}\psi \| \\
	& \quad + \sum_{j = 0}^\infty \|e^{yP_2}e^{-yP_1}E_1^{(j)}\psi\|^2.
	\end{align*}
Since $E_1[(0,a)] \leq E_2[(0,a)]$ we know that $\|e^{yP_2}E_1[(0,a)]\phi\| \leq e^{ya}\|E_1[(0,a)]\phi\|$ for all $a > 0$, so
	\begin{align*}
	\|e^{yP_2}e^{-yP_1}\psi\|^2 \leq & \, 2\sum_{j=0}^\infty \sum_{n=1}^\infty e^{2y(j+1)\delta} \|E_1^{(j)}e^{-yP_1}\psi\|\|E_1^{(j+n)}e^{-yP_1}\psi \| \\
	& \quad + \sum_{j = 0}^\infty e^{2y(j+1)\delta}\|E_1^{(j)}e^{-yP_1}\psi\|^2 \\
	\leq & \, 2 \sum_{j=0}^\infty  \sum_{n=1}^\infty e^{2y(j+1)\delta}e^{-yj\delta}e^{-y(j+n)\delta} \|E_1^{(j)}\psi\|\|E_1^{(j+n)}\psi \| \\
	& \quad + \sum_{j = 0}^\infty e^{2y(j+1)\delta}e^{-2yj\delta}\|E_1^{(j)}\psi\|^2.
	\end{align*}
Applying Cauchy-Schwarz again we get for all $n \in \N$ that
	\[ \sum_{j=0}^\infty \|E_1^{(j)}\psi\| \|E_1^{(j+n)}\psi\| \leq \left( \sum_{j=0}^\infty \|E_1^{(j)}\psi \|^2 \right)^\frac{1}{2} \left( \sum_{j=n}^\infty \|E_1^{(j)}\psi \|^2 \right)^\frac{1}{2} \leq \|\psi\|^2 \]
so that by collecting factors above we conclude
	\begin{align*}
	\|e^{yP_2}e^{-yP_1}\psi\|^2 & \leq 2 \sum_{n=1}^\infty e^{(2 - n)y\delta} \|\psi\|^2 + e^{2y\delta} \|\psi\|^2 \\
	& = e^{2y\delta} \left(1 + 2 \sum_{n=1}^\infty e^{-ny\delta} \right) \|\psi\|^2 = e^{2y\delta}\frac{1 + e^{-y\delta}}{1 - e^{-y\delta}}\|\psi\|^2.
	\end{align*}
So we see that $\|e^{yP_2}e^{-yP_1}\| \leq e^{y\delta}\sqrt{\frac{1 + e^{-y\delta}}{1- e^{-y\delta}}}$, and choosing $\delta = \frac{1}{y}$ gives a uniform bound. 

Now that we know that $z \mapsto U_2(-z)U_1(z)$ is uniformly bounded, we note that for all $\varphi, \psi \in \Hil$ the function $z \mapsto \langle \varphi, U_2(-e^{2\pi z})U_1(e^{2\pi z}) \psi \rangle$ is a bounded analytic function on $\mathbb{S}_{\pi}$, bounded by $\|\varphi\|\|\psi\|$ on the edges of the strip; by the Hadamard three lines theorem, this means that $|\langle \varphi, U_2(-z)U_1(z) \psi \rangle | \leq \|\varphi\|\|\psi\|$, from which the uniform bound follows.
\end{proof}
\begin{theorem}
\label{thm:CtoStd}
Let $H \subset \Hil$ be a standard subspace, and $(H,U_1)$ and $(H,U_2)$ be non-degenerate standard pairs. If $E_1[(0,1)] \leq E_2[(0,1)]$, then $U_1(1) H \subset U_2(1)H$.
\end{theorem}
\begin{proof}
 Let $P_j$ be the generator of $U_j$. We assume that $E_1[(0,1)] \leq E_2[(0,1)]$, and verify the condition in Lemma \ref{lem:StdPairExtension} to show that $U_1(1)H \subset U_2(1)H$.  
 
Because $\bigcup_{a \geq 0} E_1[(0,a)]$ are analytic vectors for both $U_1$ and $U_2$, we note that for all $\psi \in \bigcup_{a \geq 0} E_1[(0,a)]$ and $\varphi \in \Hil$ we have that
	\[z \mapsto \langle \varphi,  U_2(-z)U_1(z) \psi\rangle = \sum_{n,m=0}^\infty \frac{z^{n+m}}{n!m!} \langle \varphi, P_2^n P_1^m \psi \rangle\]
can be expressed as a convergent power series, so it is an analytic function.

Now let $\varphi \in \Hil$, and let $\psi_n \rightarrow \varphi$ with $\psi_n \in  \bigcup_{a \geq 0} E_1[(0,a)]$ for all $n \in \N$. Then, because by Proposition \ref{prop:ProjImpliesUnifBdd} the map $z \mapsto U_2(-z)U_1(z)$ is uniformly bounded by 1, we have
	\[ |\langle \phi, U_2(-z)U_1(z)\psi_n \rangle - \langle \phi, U_2(-z)U_1(z) \varphi \rangle| \leq \|\phi\|\|\psi_n - \varphi \| \]
so that $\langle \phi, U_2(-z)U_1(z)\psi_n \rangle \rightarrow \langle \phi, U_2(-z)U_1(z) \varphi \rangle$ as $n \rightarrow \infty$ uniformly in $z$. Since the uniform limit of analytic functions is again analytic, we have that $z \mapsto \langle \phi, U_2(-z)U_1(z) \varphi \rangle$ is analytic in $\C_+$ for all $\phi, \varphi \in \Hil$. By \cite[Sec. 3.9]{hille1948functional}, this means that $z \mapsto U_2(-z)U_1(z)$ is analytic in $\C_+$.

We can do something similar for the so-continuity: for $\psi \in \bigcup_{a \geq 0} E_1[(0,a)]$ the map $z \mapsto U_2(-z)U_1(z)\psi$ is continuous in $\overline{\C_+}$ by the functional calculus; by taking the limit $\psi_n \rightarrow \varphi$ for all $\varphi \in \Hil$, we see that $z \mapsto U_2(-z)U_1(z)\varphi$ is continuous for all $\varphi \in \Hil$. We have therefore verified the condition in Lemma \ref{lem:StdPairExtension}, concluding the proof.
\end{proof}
\begin{corollary}
\label{cor:GenIneq}
Let $(H,U_1) \subset (H,U_2)$ be an inclusion of standard pairs, and $P_i$ be the generator of $U_i$. Then $P_2 \leq P_1$.
\end{corollary}
\begin{proof}
By Proposition \ref{prop:ProjImpliesUnifBdd} we see that $\|U_2(-z)U_1(z)\| \leq 1$. This means
	\[ (U_2(-\tfrac{i}{2})U_1(\tfrac{i}{2}))^*U_2(-\tfrac{i}{2})U_1(\tfrac{i}{2}) = e^{-\frac{1}{2}P_1} e^{P_2} e^{-\frac{1}{2}P_1} \leq 1\]
so $e^{P_2} \leq e^{P_1}$. Since the logarithm is operator monotone we get $P_2 \leq P_1$.
\end{proof}
Corollary \ref{cor:GenIneq} might give the impression that we could also use the inequality of the generators of the one-parameter groups of standard pairs to characterize their inclusion. However, the concrete characterizations in Section \ref{sec:applications} allow us to construct a counterexample, as we will do in Example \ref{ex:GenIneqCounter}.

\section{Characterization in terms of inner functions and concrete examples}
\label{sec:applications}
In Theorem \ref{thm:main}, we have reduced an inclusion of standard subspaces to an inclusion of spectral subspaces. Given a standard pair or half-sided modular inclusion, such a spectral subspace might be hard to calculate explicitly. However, using the representation theory of standard pairs, we can describe these inclusions more concretely, which allows us to construct examples. Specifically, it will turn out that inclusions of standard pairs and half-sided modular inclusions are classified by sets of analytic functions. For this, we note that the set of bounded analytic functions on $\C_-$ with values in $B(\Kil)$ forms the Hardy space $H^\infty(\C_-, B(\Kil))$. Any bounded analytic function $\varphi$ on $\C_-$ has a boundary value $\lim_{y \rightarrow 0^-} \varphi(\lambda + iy)$ (with respect to the weak operator topology) for almost all $\lambda \in \R$ (this follows from the scalar-valued case, which can for example be found in \cite[Thm. 3.1]{GarnettBounded}). We will therefore identify $H^\infty(\C_-, B(\Kil)) \subset L^\infty(\R, B(\Kil))$ as multiplication operators on $L^2(\R, \Kil)$.
\begin{definition}
Let $\Kil$ be a Hilbert space. A bounded analytic function $\varphi: \C_- \rightarrow B(\Kil)$ is called an \textbf{inner function} if $\varphi(\lambda) := \text{w-}\lim_{y \rightarrow 0^-} \varphi(\lambda + iy)$ is unitary for almost all $\lambda \in \R$. If $J_\Kil$ is an anti-linear involution on $\Kil$, then $\varphi$ is called $J_\Kil$\textbf{-symmetric} if $\varphi(-\lambda) = J_\Kil \varphi(\lambda) J_\Kil$ almost everywhere. If $\Kil = \C$, we call $\varphi$ \textbf{symmetric} if $\varphi(-\lambda) = \overline{\varphi(\lambda)}$ almost everywhere.
\end{definition}

\begin{proposition}
\label{prop:IrrSymInnChar}
Let $(H,U_1)$ and $(H,U_2)$ be non-degenerate standard pairs. Then the CCR-representation generated by $\Delta_H^{it}$ and $P_1^{is}$ is irreducible if and only if the one generated by $\Delta_H^{it}$ and $P_2^{is}$ is. If this is the case, then $U_1(1)H \subset U_2(1)H$ if and only if
	\[ U_1(s) = \varphi(\ln \Delta_H) U_2(s) \varphi(\ln \Delta_H)^*  \]
for a symmetric inner function $\varphi: \C_- \rightarrow \C$.
\end{proposition}
\begin{proof}
By Corollary \ref{cor:multiplicity} the irreducibility of both representations is equivalent. If the irreducibility condition is fulfilled, by Theorem \ref{thm:SvN} the two representations are unitarily equivalent, so there exists unitary $\Vint: \Hil \rightarrow \Hil$ such that
	\[ \Vint \Delta_H^{it} \Vint^* = \Delta_H^{it}, \quad \Vint J_H \Vint^* = J_H, \quad \Vint U_2(s) \Vint^* = U_1(s).  \]
As noted above, the irreducibility condition implies that $\{ \Delta_H^{it} \mid t \in \R\}''$ is maximally abelian, so any operator that commutes with all $\Delta_H^{it}$ can be written as a function of $\ln \Delta_H$. So there is a function $\varphi: \R \rightarrow \C$ such that $\Vint = \varphi(\ln \Delta_H)$, and because $\Vint$ commutes with $J_H$ and $J_H\ln \Delta_H J_H = -\ln \Delta_H$, we have
	\[ \varphi(\ln \Delta_H) = J_H \varphi(\ln \Delta_H) J_H = \overline{\varphi}(-\ln \Delta_H). \]
Because $\Vint$ is unitary and the spectrum of $\ln \Delta_H$ is the entire real line, we must have $\overline{\varphi(\lambda)}\varphi(\lambda) = 1$ for almost all $\lambda \in \R$. So what remains to be shown, is that $\varphi$ extends to a bounded analytic function on $\C_-$.

 According to Theorem \ref{thm:main} we now have that $U_1(1)H \subset U_2(1)H$ if and only if
	\[ E_1[(0,1)] = \Vint E_2[(0,1)] \Vint^* \leq E_2[(0,1)], \]
or in other words, $\varphi(\ln \Delta_H) E_2[(0,1)] \Hil \subset E_2[(0,1)]\Hil$. According to Theorem \ref{thm:SvN} and Example \ref{ex:ex2} there is a unitary $\Wint: \Hil \rightarrow L^2(\R)$ such that 
	\begin{align*}
	\Wint \Delta_H^{it} \Wint^* & = \Delta_{H_0}^{it} = \FT^* M[e^{-2\pi it \lambda}] \FT \\
	\Wint U_2(s) \Wint^* & = U_0(s) = M[e^{ise^\lambda}].
	\end{align*}
This means that
	\[ \Wint \ln \Delta_H \Wint^* = \FT^* M[-2\pi \lambda] \FT \quad \text{and} \quad \Wint E_2[(0,1)] \Wint^* = M[\chi_{(-\infty,0)}]. \]
Applying $\Wint$ to both sides of the inclusion $\varphi(\ln \Delta_H) E_2[(0,1)] \Hil \subset E_2[(0,1)]\Hil$ we see that the inclusion holds if and only if $M[\varphi(-2\pi \lambda)]\FT L^2(\R_-) \subset \FT L^2(\R_-)$. Because $\FT L^2(\R_-) = H^2(\C_+)$, this is equivalent to $\lambda \mapsto \varphi(-2\pi \lambda)$ being in $H^\infty(\C_+)$, or in other words, $\varphi \in H^\infty(\C_-)$.
\end{proof}
\begin{corollary}
\label{cor:HSMIIrredVersion}
Let $H \subset \Hil$ be a standard subspace such that $\{\Delta_{H}^{it} \mid t \in \R\}''$ is maximally abelian, and let $K_1 \subset H$ and $K_2 \subset H$ be non-degenerate half-sided modular inclusions. Then $K_1 \subset K_2$ if and only if there is a symmetric inner function $\varphi$ on $\C_-$ such that $K_1 = \varphi(\ln\Delta_H)K_2$.
\end{corollary}
\begin{proof}
Since $K_1 \subset H$ and $K_2 \subset H$ are non-degenerate half-sided modular inclusions, by Proposition \ref{prop:nondeg} the corresponding standard pairs $(U_1, H)$ and $(U_2,H)$ are non-degenerate, and because $\{\Delta_{H}^{it} \mid t \in \R\}''$ is maximally abelian, both representations are irreducible. By Proposition \ref{prop:IrrSymInnChar} we then have
	\[ K_1 = U_1(1) \subset U_2(1) = K_2 \]
if and only if $U_1(s) = \varphi(\ln \Delta_H)U_2(s)\varphi(\ln \Delta_H)^*$ for some symmetric inner function $\varphi: \C_- \rightarrow \C$. Because $\varphi(\ln \Delta_H)$ is unitary and commutes with both $\Delta_H^{it}$ and $J_H$, we have $\varphi(\ln \Delta_H)H = H$, so if $K_1 \subset K_2$ then
	\[ K_1 = U_1(1)H = \varphi(\ln \Delta_H) U_2(1) \varphi(\ln \Delta_H)^* H = \varphi(\ln \Delta_H) K_2.  \]
Conversely, if $K_1 = \varphi(\ln \Delta_H) K_2$, then $\Delta_{K_1}^{it} = \varphi(\ln \Delta_H)\Delta_{K_2}^{it} \varphi(\ln \Delta_H)^*$, and because of Theorem \ref{thm:Wiesbrock} this implies $U_1(s) = \varphi(\ln \Delta_H)U_2(s)\varphi(\ln \Delta_H)^*$.
\end{proof}
We note that all the conditions in Corollary \ref{cor:HSMIIrredVersion} only explicitly refer to $\Delta_H^{it}$ and its relation to $H_1$ and $H_2$; in particular, one doesn't need to explicitly calculate $\Delta_{K_1}^{it}$ and $\Delta_{K_2}^{it}$, which is often not straightforward as it involves calculating a polar decomposition.
\begin{example}
Given a standard pair $(H,U)$, we consider $K_1 = U(1)H$ and $K_2 = U(\frac{1}{2})H$, which clearly form an inclusion $K_1 \subset K_2$. In fact, $(H,U)$ is the standard pair induced by the half-sided modular inclusion $K_1 \subset H$, and the standard pair associated to $K_2 \subset H$ has $s \mapsto U(s/2)$ as its positively generated one-parameter group. We note that $U(s) = \Delta_H^{-i\ln(2)/2\pi}U(\tfrac{s}{2}) \Delta_H^{i\ln(2)/2\pi}$, so for $\varphi(z) = e^{-i \ln(2)z/2\pi}$, which is indeed a symmetric inner function on $\C_-$, we have $U(s) = \varphi(\ln \Delta_{H})U(\tfrac{s}{2})\varphi(\ln \Delta_H)^*$.
\end{example}
\begin{example}
\label{ex:nontriv}
We take a simple example of a symmetric inner function on $\C_-$, namely the Blaschke product factor $\varphi(\lambda) = \frac{\lambda + i}{\lambda - i}$. Recall the definition of $H_0$ from Example \ref{ex:ex1}. Then 
	\[ \varphi(\ln\Delta_{H_0}) = \varphi(2 \pi i \partial_\theta) = \varphi(\FT M[2\pi \lambda] \FT^*) = \FT M[\varphi(2\pi \lambda)] \FT^*   \]
so as an operator, $\varphi(\ln\Delta_{H_0})$ is convolution by the Fourier transform of $\lambda \mapsto \frac{1}{\sqrt{2\pi}} \varphi(2 \pi \lambda)$. This is easily calculated to be
	\[ \frac{1}{\sqrt{2\pi}} \int_{-\infty}^\infty \frac{1}{\sqrt{2\pi}} \frac{2\pi \lambda + i}{2\pi \lambda - i}e^{-i\theta\lambda} \, d\lambda = \delta(\theta) - \frac{1}{\pi} e^{\theta/2\pi} \chi_{(-\infty,0)}(\theta)\]
To compare with Corollary \ref{cor:HSMIIrredVersion}, we set $K_2 = M[e^{ie^{\theta}}]H_0$, and $K_1 = \varphi(\ln\Delta_{H_0})K_2$. Then if $h \in H_0$, we look at
	\begin{align*}
	(\varphi(\ln\Delta_{H_0}) M[e^{ie^{\theta}}] h) (\theta) & = e^{ie^{\theta}}h(\theta) - \frac{1}{\pi}\int_{-\infty}^0 e^{\theta'/2\pi} e^{ie^{\theta - \theta'}}h(\theta-\theta') \, d\theta' \\
	& = e^{ie^{\theta}}\left(h(\theta) - \frac{1}{\pi}\int_{-\infty}^0 e^{\theta'/2\pi} e^{ie^\theta(e^{-\theta'} - 1)}h(\theta-\theta') \, d\theta'  \right).
	\end{align*}
The proposition now claims that
	\[ \widetilde{h}(\theta) := h(\theta) - \frac{1}{\pi}\int_{-\infty}^0 e^{\theta'/2\pi} e^{ie^\theta(e^{-\theta'} - 1)}h(\theta-\theta') \, d\theta' \]
is included in $H_0$, i.e. it extends analytically to the strip $\mathbb{S}_{\pi}$, is $L^2$ along lines $\im\,z = c$ for $0 < c < \pi$, and satisfies $\widetilde{h}(\theta + i\pi) = \overline{\widetilde{h}(\theta)}$. We see by the Minkwoski integral inequality that indeed
	\begin{align*}
	\|\widetilde{h}(\cdot + i\lambda)\|_2 & = \left( \int_{\R} \left| \int_{-\infty}^0 e^{\theta'/{2\pi}} e^{i e^{\theta + i\lambda}(e^{-\theta'} - 1)}h(\theta - \theta' + i\lambda) \, d\theta'  \right|^2 \, d\theta \right)^{\frac{1}{2}} \\
	& \leq \int_{-\infty}^0 \left( \int_\R |e^{\theta'/{2\pi}} e^{i e^{\theta + i\lambda}(e^{-\theta'} - 1)}h(\theta - \theta' + i\lambda) |^2 \, d\theta \right)^\frac{1}{2} \, d\theta' \\
	& = \int_{-\infty}^0 |e^{\theta'/2\pi}| \left( \int_\R |e^{-\sin(\lambda)e^{\theta}(e^{-\theta'} - 1)} h(\theta - \theta'+i\lambda) |^2 \, d\theta \right)^{\frac{1}{2}} \, d \theta'
	\end{align*}
and when $\theta' \leq 0$ and $0 \leq \lambda \leq \pi$ we have $\sin(\lambda)e^\theta(e^{-\theta'}-1) \geq 0$ so
	\begin{align*}
	\|\widetilde{h}(\cdot + i\lambda)\|_2 & \leq  \int_{-\infty}^0 |e^{\theta'/2\pi}| \left( \int_\R |h(\theta - \theta'+i\lambda) |^2 \, d\theta \right)^{\frac{1}{2}} \, d \theta' \\
	& = \|h(\cdot + i \lambda)\|_2 \int_{-\infty}^0 e^{\theta'/2\pi} d\theta' < \infty
	\end{align*}
so $\widetilde{h} \in \dom\, \Delta_{H_0}^{\frac{1}{2}}$. The identity $\widetilde{h}(\theta + i\pi) = \overline{\widetilde{h}(\theta)}$ is easily explicitly verified.
\end{example}
\begin{example}
\label{ex:GenIneqCounter}
We now examine $K_1 = \widetilde{U}_0(1)\widetilde{H}_0 \subset \widetilde{H}_0$ and $K_2 = \varphi(\ln \Delta_{\widetilde{H}_0})K_1$ for $\varphi(\lambda) = e^{i \sinh(\lambda/2)}$.

We first note that $\varphi$ is (up to possibly a factor of $-1$) the unique function $\R \rightarrow \C$ such that $\varphi(\lambda) = \overline{\varphi(-\lambda)} = \varphi(\lambda)^{-1}$ and $\varphi(\ln \Delta_{\widetilde{H}_0})K_1 = K_2$. This is because if there is another such function $\varphi'$ then
	\[ \overline{\varphi'}(\ln \Delta_{\widetilde{H}_0})\varphi(\ln \Delta_{\widetilde{H}_0})K_1 = K_1\]
so then $\overline{\varphi'}(\ln \Delta_{\widetilde{H}_0})\varphi(\ln \Delta_{\widetilde{H}_0})$ commutes with both $\Delta_{\widetilde{H}_0}^{it}$ and $\Delta_{K_1}^{it}$, meaning that it commutes with $\widetilde{U}_0(s)$ for all $s \in \R$. By irreducibility of the representation induced by $\Delta_{\widetilde{H}_0}^{it}$ and $\widetilde{U}_0(s)$, this means it is equal to a unitary constant, and because it commutes with $J_{\widetilde{H}_0}$ it must be real. Since $\overline{\varphi}$ does not have a bounded analytic extension to $\C_-$, by Corollary \ref{cor:HSMIIrredVersion} we have that $K_1 \not\subset K_2$. 

However, by Example \ref{ex:ex2ctd2} we also see that
	\begin{align*}
	\langle \psi, P_1 \psi \rangle - \langle \psi, P_2 \psi \rangle & =  \langle \psi, P_1 \psi \rangle - \langle \psi, \varphi(\ln \Delta_{\widetilde{H}_0})P_1\varphi(\ln \Delta_{\widetilde{H}_0})^* \psi \rangle  \\
	& = \| \psi( \cdot + i/2) \|_2^2 - \|\varphi(\ln \Delta_{\widetilde{H}_0})^*\psi(\cdot + i/2)\|_2^2.
	\end{align*}
So because
	\[\varphi(\ln \Delta_{\widetilde{H}_0})^*\psi(\lambda ) = \overline{\varphi(-2\pi \lambda)}\psi(\lambda) = e^{i \sinh(\pi \lambda)}\psi(\lambda) \]
we have
	\[ |\varphi(\ln \Delta_{\widetilde{H}_0})^*\psi(\lambda + i/2)| = e^{-\cosh(\pi \lambda)}|\psi(\lambda + i/2)| \leq |\psi(\lambda + i/2)| \]
so indeed $P_2 \leq P_1$. This means we have constructed a counter example to the reverse implication of Corollary \ref{cor:GenIneq}.
\end{example}

When the representation induced by the standard pair or half-sided modular inclusion is not irreducible, we can not formulate the characterization of their inclusions in such a representation-independent way as in Corollary \ref{cor:HSMIIrredVersion}. We can however still classify the inclusions of standard pairs by inner functions up to unitary equivalence. 
\begin{proposition}
\label{prop:SymInnChar}
Let $(H,U_1)$ and $(H,U_2)$ be non-degenerate standard pairs. Then $U_1(1)H \subset U_2(1)H$ if and only if there exists a Hilbert space $\Kil$, an anti-unitary involution $J_{\Kil}$ on $\Kil$, a unitary $\Wint: \Hil \rightarrow L^2(\R, \Kil)$ and a $J_\Kil$-symmetric inner function $V: \C_- \rightarrow B(\Kil)$, such that
	\begin{align*}
	\left( \Wint \Delta_{H}^{it} \Wint^* \vec{\psi}\right)(\lambda) = e^{-2\pi t \lambda i} \vec{\psi}(\lambda), & \quad \left( \Wint J_{H} \Wint^* \vec{\psi} \right)(\lambda) = J_\Kil \vec{\psi}(-\lambda)  \\
	\Wint U_1(s) \Wint^* \vec{\psi} = V(T_s * (V^* \vec{\psi})), & \qquad \Wint U_2(s) \Wint^* \vec{\psi}  = T_s * \vec{\psi}
	 \end{align*}
for all $\vec{\psi} \in L^2(\R, \Kil)$, where $T_s$ is defined as in Example \ref{ex:ex2ctd}.
\end{proposition}
\begin{proof}
By Theorem \ref{thm:SvN} and Corollary \ref{cor:multiplicity} we can find a Hilbert space $\Kil$, an anti-linear involution $J_\Kil$ on $\Kil$ and  $\Vint_1, \Vint_2: \Hil \rightarrow L^2(\R) \otimes \Kil$ such that
	\[ \Vint_j \Delta_H^{it} \Vint_j^* = \Delta_{H_0}^{it} \otimes 1_\Kil, \quad \Vint_j J_H \Vint_j^* = J_{H_0} \otimes J_\Kil ,\quad \Vint_j U_j(s) \Vint_j^* = U_0(s) \otimes 1_\Kil  \]
for $j =1,2$. In particular, we have that $U_1(s) = \Vint_1^*\Vint_2U_2(s) \Vint_2^*\Vint_1$.

By Theorem \ref{thm:main} we see that $U_1(1)H \subset U_2(1)H$ if and only if
	\[ E_1[(0,1)] = \Vint_1^*\Vint_2 E_2[(0,1)] \Vint_2^*\Vint_1 \leq E_2[(0,1)]. \]
Since we know that the spectral projection of the generator of $U_0$ is equal to multiplication with $\chi_{(-\infty,0)}$, we see that $U_1(1)H \subset U_2(1)H$ if and only if
	\[ \Vint_2\Vint_1^* \left( M[\chi_{(-\infty,0)}] \otimes 1_\Kil \right)\Vint_1\Vint_2^* \leq M[\chi_{(-\infty,0)}] \otimes 1_\Kil. \]
Identifying $L^2(\R) \otimes \Kil \cong L^2(\R,\Kil)$ we can reformulate this as $\Vint_2\Vint_1^* L^2(\R_-, \Kil) \subset L^2(\R_-, \Kil)$.

We also know that $\Delta_{\widetilde{H}_0}^{it} = \FT \Delta_{H_0}^{it} \FT^* = M[e^{-2\pi t \lambda i}]$. Because $ \Vint_2\Vint_1^* $ commutes with $\Delta_{H_0}^{it} \otimes 1_\Kil$ we see that $V := (\FT \otimes 1_\Kil) \Vint_1\Vint_2^* (\FT^* \otimes 1_\Kil)$ commutes with $\Delta_{\widetilde{H}_0}^{it} \otimes 1_\Kil$. The operator $V$ is therefore given by multiplication with a function $V: \R \rightarrow B(\Hil)$. Because $V$ is unitary, it actually takes values in $\Uni(\Hil)$, and because it commutes with $J_{\tilde{H}_0} \otimes J_\Kil$, it satisfies $J_\Kil V(\lambda) J_\Kil = V(-\lambda)$.

So since $U_1(1)H \subset U_2(1)H$ if and only if $V_2V_1^* L^2(\R_-,\Kil) \subset L^2(\R_-,\Kil)$, and because $\FT^* \otimes 1_\Kil L^2(\R_-,\Kil) = H^2(\C_+, \Kil)$, we also have $U_1(1)H \subset U_2(1)H$ if and only if $V^*H^2(\C_+,\Kil) \subset H^2(\C_+,\Kil)$, i.e. by \cite[Lem. B.1.8 (d)]{Schober2024} $V$ is a $J_\Kil$-symmetric inner function on $\C_-$. 

Setting now $\Wint := (\FT\otimes 1_\Kil)\Vint_2$, we see that 
	\begin{align*}
	\Wint U_2(s) \Wint^* & = \FT U_0(s)\FT^* \otimes 1_\Kil \\
	\Wint U_1(s) \Wint^* & = \Wint \Vint_1^* \Vint_2 U_2(s) \Vint_2^*\Vint_1 \Wint^* \\
	&= V \Wint U_2(s) \Wint^* V^*
	\end{align*}
which proves the result.
\end{proof}

\begin{remark}
\label{rem:LW}
In \cite{Longo2010}, endomorphisms of standard pairs $(H,U)$ are considered: these are unitary maps $V$ such that $VH \subset H$ and $[U(s),V] = 0$ for all $s \in \R$. It is then shown that for irreducible standard pairs, all endomorphisms are of the form $\psi(\ln P)$ for $\psi$ a symmetric inner function on the strip $\mathbb{S}_\pi$ (here $P$ is the positive generator of $U$). Clearly, this result has a similar structure to Proposition \ref{prop:IrrSymInnChar}. An extension to the reducible case is proven, as in our Proposition \ref{prop:SymInnChar}, with the main difference that in \cite{Longo2010} the result is formulated in terms of matrix elements, instead of in terms of vector-valued functions as we have done. To properly compare the results, we note that in the CCR-relation induced by the standard pair, the operators $\ln \Delta_H$ and $\ln P$ play a symmetric role (in the Schrödinger representation they are intertwined by the Fourier transform). The most crucial difference between the characterizations in this paper and those in \cite{Longo2010} is then that in the case of endomorphisms of standard pairs, one must consider symmetric inner functions on the strip (and apply them to $\ln P$), and in our case one must consider symmetric inner functions on the half-plane (and apply them to $\ln \Delta_H$). This also means that a correspondence between endomorphisms of standard pairs and inclusions $U_1(1)H \subset U_2(1)H$ of standard pairs $(H,U_1)$ and $(H,U_2)$ that relates these two results is not possible, as $\psi \circ \ln$ is a symmetric inner function on a half plane, but $P$ is a positive operator, which can never by the logarithm of a modular operator of a standard subspace.
\end{remark}

\appendix

\section{Distributional Fourier transform of $e^{ise^\theta}$}
\label{app:distroFourier}
In this appendix we prove that the distributional Fourier transform of $\theta \mapsto e^{ise^\theta}$ for $s \in \R\setminus\{0\}$ is equal to
	\[ \sqrt{\frac{\pi}{2}} \delta + \frac{1}{\sqrt{2\pi}}\mathcal{P}\left( e^{i\lambda\ln(-is)}\Gamma( - i\lambda)\right) \]
where $\lambda$ denotes the variable in Fourier space and  $\mathcal{P}$ denotes the principal value at $\lambda = 0$. This then means that
	\[ \FT M[e^{ise^\theta}] \FT^* = \frac{1}{\sqrt{2\pi}} C\left[\FT[e^{ise^\theta}] \right] = C[T_s] \]
with
	\[ T_s :=  \frac{1}{2}\delta + \frac{1}{2\pi} \PV\left(e^{i\lambda \ln(-is)}\Gamma(-i\lambda)\right) \]
One notes that for $\re\,z_1 < 0$ and $\re\,z_2 > 0$ the function $\theta \mapsto e^{z_1e^{\theta}}e^{z_2\theta}$ is Lebesgue integrable, and we have
\begin{align*}
\int_{-\infty}^\infty e^{z_1e^\theta} e^{z_2 \theta} \, d\theta & = \int_0^\infty e^{z_1\xi} e^{(z_2 - 1)\ln(\xi)} \, d\xi \\
& = \int_0^\infty e^{-\xi'} e^{(z_2 - 1)\ln(-\xi'/z_1)} \frac{1}{-z_1} \, d\xi'
\end{align*}
(here $\ln$ is the branch of the logarithm with branch cut on $\R_-$ and agreeing with the real logarithm on $\R_+$). From this we calculate
\begin{align*}
\int_{-\infty}^\infty e^{z_1e^\theta} e^{z_2 \theta} \, d\theta & = \int_0^\infty e^{-\xi'} \frac{ e^{(z_2 - 1)\ln(\xi')}}{e^{(z_2 - 1)\ln(-z_1)}} \frac{1}{-z_1} \, d\xi' \\
& = e^{-z_2\ln(-z_1)} \Gamma(z_2)
\end{align*}
For $c_1, c_2 > 0$ we define $f_{s,c_1,c_2}(\theta) := e^{sie^\theta} e^{-c_1e^{\theta}} e^{c_2 \theta}$, which is Schwartz, and we have
	\begin{align*}
	\FT[f_{s,c_1,c_2}](\lambda) & = \frac{1}{\sqrt{2\pi}}\int_{\R} e^{(-c_1 + is)e^\theta} e^{(c_2 - i\lambda)\theta} \, d\theta \\
	& = \frac{1}{\sqrt{2\pi}} e^{-(c_2 - i\lambda)\ln(c_1 - is)}\Gamma(c_2 - i\lambda).
	\end{align*}

We first consider the limit $c_1 \rightarrow 0$. We see that
	\begin{align*}
	|\FT[f_{s,c_1,c_2}](\lambda)| & = \frac{1}{\sqrt{2\pi}} e^{- c_2 \ln|c_1 - is| - \lambda \textup{arg}(c_1 - is)} |\Gamma(c_2 - i\lambda)| \\
	& \leq \frac{e^{-c_2 \ln|s|}}{\sqrt{2\pi}} e^{|\lambda|\pi/2} |\Gamma(c_2 - i\lambda)|.
	\end{align*}
Because we have
	\[ |\Gamma(i\lambda)|^2 = \frac{\pi}{\lambda \sinh(\pi \lambda)}, \quad |\Gamma(1 + i\lambda)|^2 = \frac{\pi \lambda}{\sinh(\pi \lambda)}. \]
we can use the Phragmen-Lindelöf principle to conclude that there exists a $C > 0$ such that
	\[ |\Gamma(z)| \leq  C e^{- \pi |\im\,z|/2} \]
for $z \in \{ z \in \C \mid  0< \re\, z < 1\}$. This means that the pointwise limit of $\FT[f_{s,c_1,c_2}]$ as $c_1 \rightarrow 0$ is bounded, so by dominated convergence we have for all Schwartz function $\psi$ that
	\begin{align*}
	\langle f_{s,0,c_2}, \FT[\psi] \rangle & = \lim_{c_1 \rightarrow 0^+} \langle f_{s,c_1,c_2}, \FT[\psi] \rangle = \lim_{c_1 \rightarrow 0^+} \langle \FT[f_{s,c_1,c_2}], \psi \rangle \\
		& = \left\langle \lim_{c_1 \rightarrow 0^+}  \FT[f_{s,c_1,c_2}], \psi \right\rangle.
	\end{align*}
where $\langle T, \cdot \rangle$ denotes evaluation of the distribution $T$. So
	\[\FT f_{s,0,c_2}(\lambda) =  \frac{1}{\sqrt{2\pi}} e^{-(c_2 - i\lambda)\ln(-is)}\Gamma(c_2 - i\lambda) \]
in a distributional sense. 

What remains is the limit $c_2 \rightarrow 0$. We note that $\lambda \mapsto \Gamma(i\lambda)$ has a pole at $\lambda = 0$. This pole has order 1 and residue 1 because $z\Gamma(z) = \Gamma(z + 1)$. This means that there is a ball $B_r(0) := \{ z \in \C \mid |z| < r\}$ for some $r > 0$ and an analytic function $F: B_r(0) \rightarrow \C$ such that $e^{-z\ln(-is)}\Gamma(z) = \frac{1}{z} + F(z)$. So we have
	\[ \FT f_{s,0,c_2}(\lambda)  = \frac{1}{\sqrt{2\pi}} \frac{1}{-i(\lambda + ic_2)} + \frac{1}{\sqrt{2\pi}} F(c_2 - i\lambda) \]
for $-i(\lambda + ic_2) \in B_r(0)$. We note that
	\[ \lim_{\varepsilon \rightarrow 0^+} \frac{1}{\lambda + i\varepsilon} = -i\pi \delta + \mathcal{P}\left(\frac{1}{\lambda} \right) \]
in a distributional sense. Restricting to $B_r(0)$ we see that 
	\begin{align*} \lim_{c_2 \rightarrow 0^+}  \FT f_{s,0,c_2}& = \sqrt{\frac{\pi}{2}} \delta + \frac{1}{\sqrt{2\pi}}\mathcal{P}\left( \frac{1}{\lambda}\right)  + \frac{1}{\sqrt{2\pi}}F(- i\lambda)\\
	& = \sqrt{\frac{\pi}{2}} \delta + \frac{1}{\sqrt{2\pi}}\mathcal{P}\left( e^{i\lambda\ln(-is)}\Gamma( - i\lambda)\right)
	\end{align*}
Because $e^{-z\ln(-is)}\Gamma(z)$ is bounded and analytic on $\{ z \in \C \mid 0 < \re\,z < 1\} \setminus B_r(0)$, this limit holds also outside of $B_r(0)$, so we can conclude that
	\[ \mathcal{F}[e^{ise^\theta}] = \lim_{c_2 \rightarrow 0^+}  \FT f_{s,0,c_2} = \sqrt{\frac{\pi}{2}} \delta + \frac{1}{\sqrt{2\pi}}\mathcal{P}\left( e^{i\lambda\ln(-is)}\Gamma( - i\lambda)\right)  \]

\section{Vector valued uniform convergence}
\label{app:PLUniform}
Here we provide the proof of Lemma \ref{lem:PLUniform}. The proof is adapted from the scalar-valued case in \cite[Cor. 1.4.5]{Boas1954}; it is a straightforward but technical application of the maximum principle and the Phragmén-Lindelöf principle.
\begin{lemma}
Let $A: \overline{\mathbb{S}_{\frac{1}{2}}} \rightarrow B(\Hil)$ be a norm-bounded so-continuous map that is analytic on $\mathbb{S}_\frac{1}{2}$ and $A_0 \in B(\Hil)$ such that
	\[ \lim_{t \rightarrow -\infty} A(t) = \lim_{t \rightarrow -\infty}  A(t+\tfrac{i}{2}) = A_0 \]
in the strong operator topology. Then $\lim_{t \rightarrow -\infty} A(t + iy) = A_0$ in the strong operator topology, uniformly in each seminorm for all $0 \leq y \leq \frac{1}{2}$.
\end{lemma}
\begin{proof}
Without loss of generality, we may assume that $A_0 = 0$. We note that the maximum principle holds also for vector valued analytic functions \cite[Thm. 3.12.1]{hille1948functional}. We write $\|A\|_\infty := \sup_{z \in \mathbb{S}_\frac{1}{2}} \|A(z)\|$. 

Let $\psi \in \Hil$ be a vector, and let $\varepsilon > 0$. There exists a $t_0$ such that for all $t < t_0$ we have $\|A(t)\psi\| \leq \varepsilon$ and $\|A(t + \frac{i}{2})\psi \| \leq \varepsilon$ (without loss of generality, we assume that $t_0 \leq 0$). 

We choose $\mu > 0$ such that for all $z \in \mathbb{S}_{\frac{1}{2}}$ with $\re\,z = t_0$ we have
	\[ \left| \frac{z}{z - \mu} \right| \|A\|_\infty \|\psi\| < \varepsilon. \]
In this way we guarantee that the analytic function $z \mapsto \frac{z}{z-\mu}A(z)\psi$ is bounded by $\varepsilon$ on the boundary of the halfstrip $\{ z \in \mathbb{S}_\frac{1}{2} \mid \re\, z < t_0\}$.

Now we consider for arbitrary $\lambda > 0$ the function $z \mapsto \frac{1}{z - \lambda}\frac{z}{z-\mu}A(z)\psi$. On the boundary of the set $\{ z \in \mathbb{S}_\frac{1}{2} \mid \re\,z < t_0\}$ this function is bounded by $\frac{\varepsilon}{|t_0 - \lambda|}$. On the other hand, for all $z \in \mathbb{S}_\frac{1}{2}$ such that
	\[ \left|\frac{1}{z - \lambda} \right| \|A\|_\infty \|\psi\| \leq \frac{\varepsilon}{|t_0 - \lambda|} \]
this bound is also satisfied. The set
	\[ \left \{ z \in \mathbb{S}_\frac{1}{2} \, \middle\vert \, \re\,z < t_0  \text{ and } \frac{\|A\|_\infty \|\psi\| |t_0 - \lambda|}{\varepsilon} \geq |z - \lambda|\right \} \]
is bounded, and on the boundary the function $z \mapsto \frac{1}{z - \lambda}\frac{z}{z-\mu}A(z)\psi$ is bounded by $\frac{\varepsilon}{|t_0 - \lambda|}$; by the maximum principle this means that
	\[ \left\|\frac{1}{z - \lambda}\frac{z}{z-\mu}A(z)\psi \right\| \leq \frac{\varepsilon}{|t_0 - \lambda|}\]
for all $z \in \mathbb{S}_\frac{1}{2}$ such that $\re\,z < t_0$.

We therefore have for all $\lambda > 0$ and $z \in \mathbb{S}_\frac{1}{2}$ such that $\re\,z < t_0$, that
	\[ \left\| \frac{z}{z-\mu} A(z)\psi \right\| \leq \frac{|z - \lambda|}{|t_0 - \lambda|} \varepsilon. \]
Taking $\lambda \rightarrow \infty$ we see that for all $z \in \mathbb{S}_\frac{1}{2}$ such that $\re\,z < t_0$ we have
	\[ \|A(z)\psi\| \leq \left| \frac{z - \mu}{z}\right| \varepsilon = \left( 1 + \frac{\mu}{|z|} \right) \varepsilon. \]
So for $z \in \mathbb{S}_\frac{1}{2}$ such that $\re\,z < t_0$ and $|z| > \mu$ we have $\|A(z)\psi\| \leq 2\varepsilon$. 

In this way we see that $A(t + iy)\psi \rightarrow 0$ as $t \rightarrow -\infty$, uniformly in $y$.
\end{proof}

\printbibliography

\end{document}